\documentclass[journal, 10pt]{IEEEtran}

\usepackage{xcolor,colortbl}
\usepackage[T1]{fontenc} 
\usepackage{tikz}
\usepackage{amssymb}
\usepackage{amsthm}
\usepackage{amsmath}
\usepackage{mathrsfs}
\usepackage{balance}
\usepackage{subcaption}
\usepackage{pifont}
\usepackage{enumitem}
\usepackage{booktabs}
\usepackage{multirow}
\usepackage{pgfplots}
\pgfplotsset{compat=newest}
\usetikzlibrary{intersections}
\usepgfplotslibrary{fillbetween}
\usepackage{graphicx}
\usepackage{glossaries}
\usepackage{array}
\usepackage{comment}
\usepackage{multirow}
\usepackage{makecell}
\usepackage{cite}
\usepackage{lipsum}
\usetikzlibrary{shapes}
\usepackage{float}
\usepackage{dblfloatfix}
\usepackage{tabstackengine}
\usepackage{verbatim}
\usepackage{textcomp}
\usepackage{titlesec}
\usepackage[bookmarks=false]{hyperref}
\hypersetup{colorlinks=true,linkcolor=black,citecolor=blue}
\usepackage[capitalise]{cleveref}

\definecolor{mygreen}{rgb}{0.7372, 0.9333, 0.4078}

\newcommand\mytextsf{\bfseries\sffamily\fontsize{9pt}{9pt}\selectfont}

\newtheorem{proposition}{Proposition}

\newtheoremstyle{myremark}
  {\topsep}
  {\topsep}
  {\itshape}
  {0pt}
  {\scshape}
  {.}
  { }
  {}

\theoremstyle{myremark}

\newtheoremstyle{myexample}
  {\topsep}
  {\topsep}
  {\itshape}
  {0pt}
  {\scshape}
  {.}
  { }
  {}

\theoremstyle{myexample}

\titlespacing*
    {\subsection}
    {1pt}
    {1ex}
    {1ex}
\titlespacing*%
    {\section}%
    {1pt}%
    {1ex}%
    {1ex}%

\makeatletter
\let\save@mathaccent\mathaccent
\newcommand*\if@single[3]{%
  \setbox0\hbox{${\mathaccent"0362{#1}}^H$}%
  \setbox2\hbox{${\mathaccent"0362{\kern0pt#1}}^H$}%
  \ifdim\ht0=\ht2 #3\else #2\fi
  }
\newcommand*\rel@kern[1]{\kern#1\dimexpr\macc@kerna}
\newcommand*\wideaccent[2]{\@ifnextchar^{{\wide@accent{#1}{#2}{0}}}{\wide@accent{#1}{#2}{1}}}
\newcommand*\wide@accent[3]{\if@single{#2}{\wide@accent@{#1}{#2}{#3}{1}}{\wide@accent@{#1}{#2}{#3}{2}}}
\newcommand*\wide@accent@[4]{%
  \begingroup
  \def\mathaccent##1##2{%
    \let\mathaccent\save@mathaccent
    \if#42 \let\macc@nucleus\first@char \fi
    \setbox\z@\hbox{$\macc@style{\macc@nucleus}_{}$}%
    \setbox\tw@\hbox{$\macc@style{\macc@nucleus}{}_{}$}%
    \dimen@\wd\tw@
    \advance\dimen@-\wd\z@
    \divide\dimen@ 3
    \@tempdima\wd\tw@
    \advance\@tempdima-\scriptspace
    \divide\@tempdima 10
    \advance\dimen@-\@tempdima
    \ifdim\dimen@>\z@ \dimen@0pt\fi
    \rel@kern{0.6}\kern-\dimen@
    \if#41
      #1{\rel@kern{-0.6}\kern\dimen@\macc@nucleus\rel@kern{0.4}\kern\dimen@}%
      \advance\dimen@0.4\dimexpr\macc@kerna
      \let\final@kern#3%
      \ifdim\dimen@<\z@ \let\final@kern1\fi
      \if\final@kern1 \kern-\dimen@\fi
    \else
      #1{\rel@kern{-0.6}\kern\dimen@#2}%
    \fi
  }%
  \macc@depth\@ne
  \let\math@bgroup\@empty \let\math@egroup\macc@set@skewchar
  \mathsurround\z@ \frozen@everymath{\mathgroup\macc@group\relax}%
  \macc@set@skewchar\relax
  \let\mathaccentV\macc@nested@a
  \if#41
    \macc@nested@a\relax111{#2}%
  \else
    \def\gobble@till@marker##1\endmarker{}%
    \futurelet\first@char\gobble@till@marker#2\endmarker
    \ifcat\noexpand\first@char A\else
      \def\first@char{}%
    \fi
    \macc@nested@a\relax111{\first@char}%
  \fi
  \endgroup
}
\makeatother

\newcommand\widebar{\wideaccent\overline}

\makeatletter
\newcommand{\doublewidetilde}[1]{{%
  \mathpalette\double@widetilde{#1}%
}}
\newcommand{\double@widetilde}[2]{%
  \sbox\z@{$\m@th#1\widetilde{#2}$}%
  \ht\z@=.9\ht\z@
  \widetilde{\box\z@}%
}
\makeatother

\newcommand\scaleddot{\scalebox{.89}{.}}

\makeatletter
\renewcommand{\dddot}[1]{%
  {\mathop{\kern\z@#1}\limits^{\makebox[0pt][c]{\vbox to-2.2\ex@{\kern-\tw@\ex@
   \hbox{\normalfont\scaleddot\kern-0.5pt\scaleddot\kern-0.5pt\scaleddot}\vss}}}}}
\renewcommand{\ddddot}[1]{%
  {\mathop{\kern\z@#1}\limits^{\makebox[0pt][c]{\vbox to-2.2\ex@{\kern-\tw@\ex@
   \hbox{\normalfont\scaleddot\kern-0.5pt\scaleddot\kern-0.5pt\scaleddot\kern-0.5pt\scaleddot}\vss}}}}}
\makeatother

\makeglossaries

\newacronym{ISAC}{ISAC}{integrated sensing and commmunications}
\newacronym{BS}{BS}{base station}
\newacronym{RF}{RF}{radio-frequency}
\newacronym{DAC}{DAC}{digital-to-analog converter}
\newacronym{RIS}{RIS}{reconfigurable intelligent surface}
\newacronym{PAPR}{PAPR}{peak-to-average power ratio}
\newacronym{PLS}{PLS}{physical-layer security}

\newacronym{AWGN}{AWGN}{additive white Gaussian noise}
\newacronym{SNR}{SNR}{signal-to-noise ratio}
\newacronym{SINR}{SINR}{signal-to-interference-plus-noise ratio}
\newacronym{SDR}{SDR}{semidefinite relaxation}
\newacronym{SDP}{SDP}{semidefinite programming}
\newacronym{SCA}{SCA}{successive convex approximation}

\newacronym{IPM}{IPM}{interior point method}  
\newacronym{MILP}{MILP}{mixed-integer linear program} 
\newacronym{MINLP}{MINLP}{mixed-integer nonlinear program} 
\newacronym{MISDP}{MISDP}{mixed-integer semidefinite program} 
\newacronym{MISOCP}{MISOCP}{mixed-integer second-order cone program} 
\newacronym{MIQCP}{MIQCP}{mixed-integer quadratically constrained program} 

\newacronym{AOA}{AOA}{angle of arrival}
\newacronym{AOD}{AOD}{angle of departure}
\newacronym{BME}{BME}{beampattern matching error}
\newacronym{DPG}{DPG}{directional power gain}
\newacronym{RC}{RC}{reflection coefficient}
\newacronym{CSI}{CSI}{channel state information}
\newacronym{MCS}{MCS}{modulation and coding scheme}
\newacronym{BLER}{BLER}{block error rate}
\newacronym{CQI}{CQI}{channel quality indicator}

\newacronym{QoS}{QoS}{quality-of-service} 
\newacronym{LoS}{LoS}{line-of-sight}
\newacronym{NLoS}{NLoS}{non-LoS}

\newacronym{LHS}{LHS}{left-hand-side}
\newacronym{RHS}{RHS}{right-hand-side}
\newacronym{THz}{THz}{terahertz}

\newacronym{BGM}{BGM}{bipartite graph matching}

\newacronym{ES}{ES}{exhaustive search}
\newacronym{BnC}{BnC}{branch-and-cut}

\newacronym{CRB}{CRB}{Cramér-Rao bound}

\begin{document}


%





\title{\huge Secure RIS-Aided Multicasting: Globally Optimal Beam Management and Discrete-Phase RIS Configuration}




\author{
\IEEEauthorblockN{Luis F. Abanto-Leon and 
				  Setareh Maghsudi} \\ 
\IEEEauthorblockA{Ruhr University Bochum, Germany \\
\{luis.abantoleon, setareh.maghsudi\}@ruhr-uni-bochum.de} 
}

\maketitle
\medskip

\begin{abstract}

\Glspl{RIS} are poised to revolutionize wireless multicasting by enabling extended coverage and reliable operation in obstructed environments. These benefits, however, can be undermined by security vulnerabilities arising from practical deployment factors. This work addresses three such critical factors, (i) the discrete nature of \gls{RIS} phase shifts, (ii) the presence of colluding eavesdroppers, and (iii) the inefficiency of static illumination beams, each threatening security if not properly accounted for in system design. To mitigate these issues, we formulate a joint resource allocation problem that minimizes the \emph{wiretap \gls{SNR}} across all eavesdroppers by co-optimizing the \gls{RIS} configuration and the \gls{BS} beam management. This yields a complex, nonconvex \gls{MINLP}, which we equivalently reformulate into a tractable \gls{MIQCP} solvable to global optimality. Numerical results confirm that the proposed scheme significantly bolsters security, suppressing the wiretap \gls{SNR} by up to $ 58\%$ compared to existing baselines.

\end{abstract}





\glsresetall
\section{Introduction} \label{sec:introduction}

To meet soaring capacity demands, 3GPP has prioritized operation in high-frequency bands, such as mmWave and THz, which offer abundant spectrum resources \cite{abanto2025:optimal-user-target-scheduling-user-target-pairing-low-resolution-phase-only-beamforming-isac-systems}. However, these frequencies are highly susceptible to blockage. To mitigate this, \glspl{RIS} have emerged as a transformative technology capable of reconfiguring wireless propagation environments to bypass obstacles. In parallel, multicasting is envisioned as a key enabler for 6G, exploiting the wide bandwidths available in high-frequency spectra \cite{chukhno2024:models-methods-solutions-multicasting-5g-6g-mmwave-sub-thz-systems}. The synergy of \glspl{RIS} and multicasting thus holds significant promise for ubiquitous, high-capacity coverage.

Despite this potential, ensuring \gls{PLS} in \gls{RIS}-aided multicast systems remains highly challenging. Existing security-oriented designs often rely on idealized assumptions, overlooking critical practical deployment factors. As discussed next, neglecting these aspects can lead to significant security vulnerabilities.

\textit{Discrete phases:} Most RIS designs assume continuous phases, which lack practical feasibility, e.g., \cite{shu2023intelligent, wang2022intelligent}. To address this, continuous-phase relaxation followed by projection has become standard practice \cite{yan2023passive, xu2024reconfigurable, du2025discrete}. However, the projection step distorts the intended RIS beampattern by introducing quantization artifacts and generating unintended sidelobes that can be exploited by eavesdroppers, representing a largely overlooked vulnerability in current research.

\textit{Colluding eavesdroppers:} While some works addressed \gls{RIS}-assisted multicasting under independent eavesdroppers \cite{wang2022multicast, lin2023secure}, practical adversaries may collude to enhance their joint decoding capability \cite{geraci2014secrecy, shu2023intelligent, wang2022intelligent}. This threat is acute in multicasting, where a common signal is broadcast to multiple users, expanding the attack surface. Yet, collusion-based interception in such settings remains largely unexplored.

\textit{Dynamic illuminating beam:}  While beam management has been studied in prior works \cite{dejonghe2024design, yilmaz2025joint}, its role in secure communications remains largely underexplored. When direct \gls{BS}-user links are obstructed, a common approach is to use a static \gls{BS}-to-\gls{RIS} illuminating beam~\cite{abanto2026:fast, deram2025risense}, consistent with fixed infrastructure. However, this imposes a fixed angle of incidence, restricting the spatial degrees of freedom at the \gls{RIS} and limiting its ability to suppress eavesdroppers, thereby motivating dynamic beam management.

\begin{table}[!t]
	\begin{center}
		\fontsize{6.5}{6}\selectfont
		\setlength\tabcolsep{2.8pt}
		\renewcommand{\arraystretch}{0.85}
		\setlength{\extrarowheight}{0pt}
		\centering
		\caption{Comparison of related work.}
		\label{table:related-literature}
		\begin{tabular}{c c c c c c} 
			\toprule
			\makecell{\textbf{Works}} & \makecell{\textbf{System}}  & \makecell{\textbf{RIS phases}} & \makecell{\textbf{Beam} \\ \textbf{management}} &  \makecell{\textbf{Eavesdroppers}} & \makecell{\textbf{Solution}} 
			\\
			\midrule
			\cite{shu2023intelligent, wang2022intelligent} & Unicast & Continuous & N/A & \multicolumn{1}{>{\columncolor{green!20}}c}{\tabularCenterstack{c}{\textbf{Colluding}}} & Suboptimal 
			\\
			\midrule
			\cite{yan2023passive} & \multicolumn{1}{>{\columncolor{green!20}}c}{\tabularCenterstack{c}{\textbf{Multicast}}} & \makecell{Discrete via projection} & Fixed & N/A & Suboptimal 
			\\
			\midrule
			\cite{xu2024reconfigurable, du2025discrete} & \multicolumn{1}{>{\columncolor{green!20}}c}{\tabularCenterstack{c}{\textbf{Multicast}}} & \makecell{Discrete via projection}  & N/A & N/A & Suboptimal 
			\\
			\midrule
			\cite{wang2022multicast, lin2023secure} & \multicolumn{1}{>{\columncolor{green!20}}c}{\tabularCenterstack{c}{\textbf{Multicast}}} & Continuous & N/A & Non-colluding & Suboptimal 
			\\
			\midrule
			\cite{dejonghe2024design, yilmaz2025joint} & Unicast &  \makecell{Discrete via projection}  & \multicolumn{1}{>{\columncolor{green!20}}c}{\tabularCenterstack{c}{\textbf{Dynamic}}} & N/A & Suboptimal
			\\
			\midrule
			\cite{abanto2026:fast} & Unicast & \multicolumn{1}{>{\columncolor{green!20}}c}{\tabularCenterstack{c}{\textbf{Discrete}}} & Fixed & N/A & Suboptimal
			\\
			\midrule
			\cite{deram2025risense} & Unicast & \makecell{Discrete via projection}  & Fixed & N/A & Suboptimal
			\\
			\midrule
			\multicolumn{1}{>{\columncolor{green!20}}c}{\tabularCenterstack{c}{\textbf{Proposed}}} & \multicolumn{1}{>{\columncolor{green!20}}c}{\tabularCenterstack{c}{\textbf{Multicast}}} & \multicolumn{1}{>{\columncolor{green!20}}c}{\tabularCenterstack{c}{\textbf{Discrete}}} & \multicolumn{1}{>{\columncolor{green!20}}c}{\tabularCenterstack{c}{\textbf{Dynamic}}} & \multicolumn{1}{>{\columncolor{green!20}}c}{\tabularCenterstack{c}{\textbf{Colluding}}} & \multicolumn{1}{>{\columncolor{green!20}}c}{\tabularCenterstack{c}{\textbf{Optimal}}}
			\\
			\bottomrule
		\end{tabular}
	\end{center}
	\vspace{-5.5mm}
\end{table}

Despite their practical relevance, the impact of these critical factors on \gls{PLS} remains largely unexamined. To bridge this gap, we formulate a comprehensive resource allocation problem that jointly optimizes (i) the \gls{RIS} configuration and (ii) the \gls{BS} beam management to minimize the wiretap \gls{SNR} under colluding eavesdroppers, while ensuring the required quality of service for legitimate users. The resulting problem is a nonconvex \gls{MINLP}, which we equivalently reformulate as a tractable \gls{MIQCP} solvable to global optimality. Numerical results demonstrate that the proposed joint design significantly suppresses the wiretap \gls{SNR} compared to existing benchmarks. A detailed comparison with the state-of-the-art is provided in \cref{table:related-literature}.

\emph{Notation}: The transpose and Hermitian transpose are denoted by $ (\cdot)^\mathrm{T} $ and $ (\cdot)^\mathrm{H} $. The sets of complex, binary, and non-negative numbers are denoted by $ \mathbb{C} $, $ \mathbb{B} $, and $ \mathbb{R}_+ $, respectively.


\section{System Model and Problem Formulation} \label{sec:system-model-problem-formulation}

We consider a downlink multicast system consisting of a multi-antenna \gls{BS}, an \gls{RIS}, $U$ single-antenna users, and $E$ single-antenna eavesdroppers. The direct BS-user links are blocked, so communication relies solely on signals reflected by the \gls{RIS}, as illustrated in \cref{fig:system-model}. The \gls{RIS} comprises $M_{\mathrm{x}} \times M_{\mathrm{z}}$ uniformly spaced elements in the $xz$-plane, with total size $M = M_{\mathrm{x}} M_{\mathrm{z}}$. The \gls{BS} has $N$ antennas and is located at the left of the \gls{RIS}, with its antennas placed along the $y$-axis. Users and eavesdroppers are located in the half-space in front of the \gls{RIS}. The corresponding index sets are $\mathcal{U} = \{1,\dots,U\}$ and $\mathcal{E} = \{1,\dots,E\}$. For notational simplicity, the $ u $-th user is denoted by $ \mathsf{U}_u $ and the $ e $-th eavesdropper by $ \mathsf{E}_e $.
 
\textbf{Beam management model:}
The \gls{BS} employs a predefined beam codebook $ \{\mathbf{t}_k\}_{k \in \mathcal{K}} $ comprising $ K $ candidate beams for \gls{RIS} illumination, such that $ \|\mathbf{t}_k\|_2^2 = P_{\mathrm{tx}} $, where $P_{\mathrm{tx}}$ is the total transmit power and $ \mathcal{K} = \{1, \dots, K\} $. To model the selection of a single beam, we introduce the following constraints 
\begin{align*}
	& \mathrm{C}_{1}: ~ \alpha_k \in \mathbb{B}, \forall k \in \mathcal{K}, \label{aa}
	~~
	& \mathrm{C}_{2}: ~ \textstyle \sum_{k \in \mathcal{K}} \alpha_k = 1.
\end{align*}

In $ \mathrm{C}_{1} $, $ \alpha_k = 1 $ indicates that beam $ \mathbf{t}_k $ is active, and $ \alpha_k = 0 $ otherwise. In addition, $ \mathrm{C}_{2} $ enforces that exactly one beam is chosen from the codebook. Consequently, the transmit beam used by the \gls{BS} is given by
\begin{align}
	\textstyle \mathbf{b} = \sum_{k \in \mathcal{K}} \alpha_k \mathbf{t}_k.
\end{align}

\textbf{RIS phase shift model:}
The \gls{RIS} is characterized by $ \mathbf{w} \in \mathbb{C}^{M \times 1} $, where each element $ \left[ \mathbf{w} \right]_m $ corresponds to the phase shift of the $m$-th reflecting element. The phase shifts are drawn from a finite discrete set, modeled as
\begin{align*} 
	\mathrm{C}_{3}: & ~ \left[ \mathbf{w} \right]_m \in \left\lbrace \mathrm{e}^{\mathrm{j} \phi_1}, \dots, \mathrm{e}^{\mathrm{j} \phi_Q} \right\rbrace, \forall m \in \mathcal{M}, 
\end{align*}
where $\phi_q$ denotes the $q$-th phase value and $Q$ is the number of available phase choices.

\textbf{Communication model:}
The \gls{BS} transmits data symbol $ s \in \mathbb{C} $, modeled as a zero-mean, unit-variance complex random variable, i.e., $ \mathbb{E} \left\lbrace s s^{*} \right\rbrace = 1 $. Thus, after reflection by the \gls{RIS}, the signal received by $ \mathsf{U}_u $ is 
\begin{align*}
	y_{u} & = \mathbf{h}^\mathrm{T}_{u} \mathrm{diag} (\mathbf{w}) \mathbf{G} \mathbf{b}  s + n_{u},
	\\
			& = \textstyle \sum_{k \in \mathcal{K}} \alpha_k \mathbf{h}^\mathrm{T}_{u} \mathrm{diag} (\mathbf{w}) \mathbf{G} \mathbf{t}_k s + n_{u},
\end{align*}
where $ \mathbf{h}_{u} \in \mathbb{C}^{M \times 1} $ is the channel between the \gls{RIS} and $ \mathsf{U}_u $, $ \mathbf{G} \in \mathbb{C}^{M \times N } $ is the channel between the \gls{RIS} and \gls{BS}, and $ n_{u} \sim \mathcal{CN} \left( 0,\sigma_{u}^2 \right) $ is \gls{AWGN} at $ \mathsf{U}_u $. Hence, the \gls{SNR} of $ \mathsf{U}_u $ is
\begin{equation}
	\mathsf{SNR}_{u} \left( \boldsymbol{\Omega} \right) = \textstyle \big|\sum_{k \in \mathcal{K}} \alpha_k \mathbf{h}^\mathrm{T}_{u} \mathrm{diag} (\mathbf{w}) \mathbf{G} \mathbf{t}_k  \big|^2 / {\sigma_{u}^2} ,
\end{equation}
where $ \boldsymbol{\Omega} \triangleq \left( \boldsymbol{\alpha}, \mathbf{w} \right) $ denotes the set of all decision variables in the resource allocation problem, with $ \boldsymbol{\alpha} = \left[ \alpha_1, \dots, \alpha_K \right]^\mathrm{T} $. To ensure reliable communication, each user is required to satisfy
\begin{align*}
	\mathrm{C}_{4}: & ~ \mathsf{SNR}_{u} \left( \boldsymbol{\Omega} \right) \geq \Gamma_\mathrm{th}, \forall u \in \mathcal{U}, 
\end{align*} 
where $ \Gamma_\mathrm{th} $ is the imposed \gls{SNR} threshold.

\textbf{Eavesdropping model:} The signal reflected by the \gls{RIS} and received by $ \mathsf{E}_e $ is
\begin{align*}
	\widebar{y}_{e} & = \mathbf{f}_{e}^\mathrm{T} \mathrm{diag} (\mathbf{w}) \mathbf{G} \mathbf{b} s + \widebar{n}_{e},
	\\
	& = \textstyle \sum_{k \in \mathcal{K}} \alpha_k \mathbf{f}_{e}^\mathrm{T} \mathrm{diag} (\mathbf{w}) \mathbf{G} \mathbf{t}_k s + \widebar{n}_{e},
\end{align*}
where $ \mathbf{f}_{e} \in \mathbb{C}^{M \times 1} $ is the channel between the \gls{RIS} and $ \mathsf{E}_e $, while $ \widebar{n}_{e} \sim \mathcal{CN} \left( 0, \widebar{\sigma}_{e}^2 \right) $ is \gls{AWGN} at $ \mathsf{E}_e $. The \gls{SNR} experienced by eavesdropper $ \mathsf{E}_e $ is
\begin{equation} \label{eqn:eavesdropper-snr}
	\widebar{\mathsf{SNR}}_{e} \left( \boldsymbol{\Omega} \right) =  \textstyle \big| \textstyle \sum_{k \in \mathcal{K}} \alpha_k \mathbf{f}_{e}^\mathrm{T} \mathrm{diag} (\mathbf{w}) \mathbf{G} \mathbf{t}_k \big|^2 / 
	{ \widebar{\sigma}_{e}^2}.
\end{equation}

Under a collusive strategy, the eavesdroppers can enhance information decoding through collaboration. Thus, we employ the \emph{wiretap \gls{SNR}} in \cite{geraci2014secrecy, wang2024secrecy} as our performance metric, which represents the collective signal interception capability of all eavesdroppers, and is defined as
\begin{equation} \label{eqn:collusive-snr}
	\widetilde{\mathsf{SNR}} \left( \boldsymbol{\Omega} \right) = \textstyle \sum_{e \in \mathcal{E}} \widebar{\mathsf{SNR}}_{e} \left( \boldsymbol{\Omega} \right).
\end{equation}

\begin{figure}[!t]
	\centering
	\includegraphics[width=0.68\columnwidth]{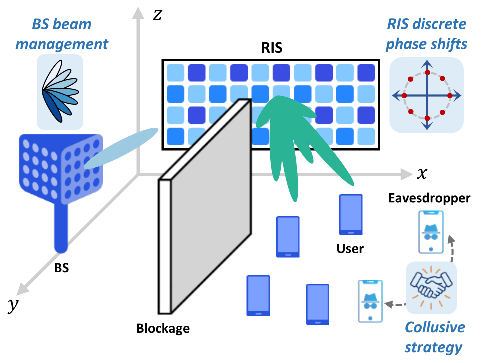}
	\vspace{-2mm}
	\caption{Secure RIS-assisted multicast system.}	
	\label{fig:system-model}
	\vspace{-5mm}
\end{figure}

\textbf{Channel model:} 
Let $ \mathbf{G} = {\omega} \Big( \sqrt{\tfrac{{\chi}}{{\chi}+1}} \cdot \mathbf{G}^\mathrm{LoS} + \sqrt{\tfrac{1}{{\chi}+1}} {\mathbf{G}}^\mathrm{NLoS} \Big) $, where $ \mathbf{G}^\mathrm{LoS} $ and $ \mathbf{G}^\mathrm{NLoS} $ are the \gls{LoS} and \gls{NLoS} components, $ {\omega} $ is the large-scale fading coefficient, and $ \chi $ is the Rician fading factor. Here, $ \mathbf{G}^\mathrm{LoS} = \mathrm{e}^{ -\mathrm{j} 2 \pi f_0 \tfrac{d}{v}}  {\mathbf{a}}_{\mathrm{ris}} (\gamma^\mathrm{az},\gamma^\mathrm{el}) ~ {\mathbf{a}}_{\mathrm{bs}}^\mathrm{T} (\gamma^\mathrm{az},\gamma^\mathrm{el}) $, where $d$ is the RIS-BS distance, $f_0$ is the carrier frequency, and $ v $ is the speed of light. The \gls{RIS} is centered at $(0,0,z_{\mathrm{ris}})$ and the \gls{BS} is centered at $(\dot{x}_{\mathrm{bs}}, \dot{y}_{\mathrm{bs}}, \dot{z}_{\mathrm{bs}})$. The $m$-th element of $ \mathbf{a}_{\mathrm{ris}} (\gamma^\mathrm{az},\gamma^\mathrm{el}) $ is $ \mathrm{e}^{\mathrm{j} \tfrac{2 \pi f_0}{v} \mathbf{v}^\mathrm{T} \mathbf{p}^{\mathrm{ris}}_m } $, where $ \mathbf{p}^{\mathrm{ris}}_m $ is its coordinate relative to the RIS and $ \mathbf{v}= \left[ \sin (\gamma^\mathrm{el}) \cos (\gamma^\mathrm{az}), \sin (\gamma^\mathrm{el}) \sin (\gamma^\mathrm{az}), \cos (\gamma^\mathrm{el}) \right]^\mathrm{T} $. Similarly, the $n$-th element of $ \mathbf{a}_{\mathrm{bs}} (\gamma^\mathrm{az},\gamma^\mathrm{el}) $ is $ \mathrm{e}^{\mathrm{j} \tfrac{2 \pi f_0}{v} \mathbf{v}^\mathrm{T} \mathbf{p}^{\mathrm{bs}}_n } $, where $ \mathbf{p}^{\mathrm{bs}}_n $ is its coordinate relative to the BS. The channels $ \mathbf{h}_{u} $ and $ \mathbf{f}_{e} $ are defined analogously, specifically, $ \mathbf{h}_{u} = \hat{\omega}_{u} \Big( \sqrt{\tfrac{\hat{\chi}_u}{\hat{\chi}_u+1}} \cdot \mathbf{h}_{u}^\mathrm{LoS} + \sqrt{\tfrac{1}{\hat{\chi}_u+1}} {\mathbf{h}}_{u}^\mathrm{NLoS} \Big) $ and $ \mathbf{f}_{e} = \check{\omega}_{e} \Big( \sqrt{\tfrac{\check{\chi}_e}{\check{\chi}_e+1}} \cdot \mathbf{f}_{e}^\mathrm{LoS} + \sqrt{\tfrac{1}{\check{\chi}_e+1}} {\mathbf{f}}_{e}^\mathrm{NLoS} \Big) $, with parameters $ \hat{\omega}_{u} $, $ \check{\omega}_{e} $, $ \hat{\chi}_u $, and $ \check{\chi}_e $. Here, $ \mathbf{h}_{u}^\mathrm{LoS} = \mathrm{e}^{ -\mathrm{j} 2 \pi f_0 \tfrac{\hat{d}_u}{v}} {\mathbf{a}}_{\mathrm{ris}} (\theta_u^\mathrm{az}, \theta_u^\mathrm{el}) $ and $ \mathbf{f}_{e}^\mathrm{LoS} = \mathrm{e}^{ -\mathrm{j} 2 \pi f_0 \tfrac{\check{d}_e}{v}} {\mathbf{a}}_{\mathrm{ris}} (\beta_e^\mathrm{az}, \beta_e^\mathrm{el}) $, where $ (\theta_u^\mathrm{az}, \theta_u^\mathrm{el}) $ and $ (\beta_e^\mathrm{az}, \beta_e^\mathrm{el}) $ are the \gls{LoS} angles of $ \mathsf{U}_u $ and $ \mathsf{E}_e $, respectively, while $ \hat{d}_u $ and $ \check{d}_e $ denote their respective distances from the \gls{RIS}. The \gls{NLoS} components $ \mathbf{h}_{u}^\mathrm{NLoS} $, $ \mathbf{f}_{e}^\mathrm{NLoS} $, and $ \mathbf{G}^\mathrm{NLoS} $ are modeled as zero-mean, unit-variance complex Gaussian random variables\footnote{We assume user \gls{CSI} is obtained via uplink pilots, whereas eavesdropper \gls{CSI} is estimated using sensing or location-based methods. To establish an optimal performance bound, we assume perfect \gls{CSI} for both users and eavesdroppers. In practice, however, \gls{CSI} is inherently imperfect, which inevitably leads to a secrecy performance degradation relative to the ideal case. While a rigorous treatment of \gls{CSI} uncertainty is crucial, it remains beyond the scope of the current work and is therefore deferred to future investigation.}.

\textbf{Problem formulation:} To safeguard multicast transmission against interception, we jointly optimize the \gls{RIS} phase configuration and the \gls{BS} beam management, ensuring a target SNR for legitimate users while minimizing the wiretap SNR. The corresponding resource allocation problem is formulated as
\begin{align*} 
	\mathcal{P}: ~~ \underset{ \boldsymbol{\Omega}}{\mathrm{minimize}} ~ f \left( \boldsymbol{\Omega} \right) \triangleq \widetilde{\mathsf{SNR}} \left( \boldsymbol{\Omega} \right)  ~~ \mathrm{s.t.} ~~ \mathrm{C}_{1} - \mathrm{C}_{4}, 
\end{align*}
where $\mathcal{P}$ is a nonconvex \gls{MINLP} and thus challenging to solve. In particular, obtaining its global optimum would require exhaustive enumeration over all \gls{RIS} phase configurations and \gls{BS} candidate beam selections, leading to a prohibitive worst-case computational complexity of $\mathcal{O}(K Q^M)$.

\section{Proposed MIQCP Reformulation} \label{sec:proposed-approach}

To circumvent the high complexity of $\mathcal{P}$, we reformulate it as an equivalent \gls{MIQCP}, denoted $\mathcal{P}'$, via \textbf{\cref{thm:proposition-1}} to \textbf{\cref{thm:proposition-4}}. These propositions transform intractable expressions into convex equivalents while preserving the original feasible set. The proofs are provided in the \textbf{Appendix}.

\subsection{Transformation of the objective function}

To handle the intractable objective function, we introduce an auxiliary variable to convert the objective into an equivalent constraint. We then decompose this constraint into elementary expressions, yielding a tractable reformulation that facilitates subsequent manipulation.
\begin{proposition} \label{thm:proposition-1}
	The objective function $ f \left( \boldsymbol{\Omega} \right) \triangleq $ $ ~ \widetilde{\mathsf{SNR}} \left( \boldsymbol{\Omega} \right) $ can be recast as a new function $ g \left( \boldsymbol{\Omega}' \right) \triangleq \mu $ by introducing additional constraints $ \mathrm{D}_{1} $, $ \mathrm{D}_{2} $, $ \mathrm{D}_{3} $, and $ \mathrm{D}_{4} $,
	\begin{equation} \nonumber
		f \left( \boldsymbol{\Omega} \right)
		\Leftrightarrow
			\begin{cases}
					g \left( \boldsymbol{\Omega}' \right) \triangleq \mu,
					~~~~~~~~~~~~~~~~~~~~
				   	\mathrm{D}_{1}: ~ \mu \in \mathbb{R}_+,
				   	\\	
				   	\mathrm{D}_{2}: ~ \delta_{e} \in \mathbb{R}_+, \forall e \in \mathcal{E}, 
				   	~~~~~~~
				   	\mathrm{D}_{3}: ~ \sum_{e \in \mathcal{E}} \delta_{e}^2 \leq \mu, 
				   	\\	
				   	\mathrm{D}_{4}: ~ \frac{\big| \textstyle \sum_{k \in \mathcal{K}} \alpha_k \mathbf{p}_{e,k}^\mathrm{H} \mathbf{w} \big| }{ \widebar{\sigma}_{e}}  \leq \delta_{e}, \forall e \in \mathcal{E}, 
			\end{cases}
	\end{equation}
	where $ \mu $ and $ \delta_{e} $ are new variables, $ \mathbf{p}_{e,k} = \mathrm{diag} \left( \mathbf{G}^* \mathbf{t}_k^* \right) \mathbf{f}_{e}^* $, and $ \boldsymbol{\Omega}' $ denotes the set of decision variables of the reformulated problem $ \mathcal{P}' $. This set is progressively augmented as additional variables are introduced during the transformation.
	
\end{proposition}

\subsection{Transformation of constraint $ \mathrm{D}_{4} $}

The summation inside the absolute value, together with the multiplicative coupling between $ \alpha_k $ and $ \mathbf{w} $, makes constraint $ \mathrm{D}_{4} $ intractable. To enable efficient optimization, we reformulate $ \mathrm{D}_{4} $ into an equivalent form that eliminates these couplings while preserving the original feasible solution set.
\begin{proposition} \label{thm:proposition-2}
	Constraint $ \mathrm{D}_{4} $ can be equivalently rewritten as constraint $ \mathrm{E}_{1} $
	\begin{align*} \nonumber
		\mathrm{D}_{5} \Leftrightarrow
				   	\mathrm{E}_{1}: & \frac{\big| \mathbf{p}_{e,k}^\mathrm{H} \mathbf{w} \big| }{ \widebar{\sigma}_{e}} \leq \delta_{e} + (1 - \alpha_k) L_{e}, \forall e \in \mathcal{E}, k \in \mathcal{K},
	\end{align*}
	where $ L_{e} = \sqrt{M P_\mathrm{tx}} \left\| \mathbf{f}_{e} \right\|_2 \left\| \mathbf{G}  \right\|_\mathrm{F} / \widebar{\sigma}_{e} $. 
\end{proposition}

\subsection{Transformation of constraint $ \mathrm{C}_{3} $}

To circumvent the combinatorial complexity of $\mathrm{C}_3$, we adopt a one-hot encoding approach. This replaces the multiple-choice selection with a set of linear constraints while preserving global optimality within the reformulated solution space.
\begin{proposition} \label{thm:proposition-3}
	Constraint $ \mathrm{C}_{3} $ can be equivalently rewritten as constraints $ \mathrm{F}_{1} $, $ \mathrm{F}_{2} $, and $ \mathrm{F}_{3} $
	\begin{equation} \nonumber
		\mathrm{C}_{3} \Leftrightarrow
			\begin{cases}
				   	\mathrm{F}_{1}: \mathbf{z}_m \in \mathbb{B}^{Q \times 1}, \forall m \in \mathcal{M}, 
				   	\\
				   	\mathrm{F}_{2}: \mathbf{1}^\mathrm{T} \mathbf{z}_m = 1, \forall m \in \mathcal{M},
				   	\\
				   	\mathrm{F}_{3}: \left[ \mathbf{w} \right]_m = \mathbf{q}^\mathrm{T} \mathbf{z}_m, \forall m \in \mathcal{M},
			\end{cases}
	\end{equation}
	where $ \mathbf{z}_m $ are new variables and $ \mathbf{q} = \left[ \mathrm{e}^{\mathrm{j} \phi_1}, \dots, \mathrm{e}^{\mathrm{j} \phi_Q} \right]^\mathrm{T} $.
\end{proposition}

\subsection{Transformation of constraint $ \mathrm{C}_{4} $}

The nonconvexity and summation over bilinear terms $\alpha_k \mathbf{w}$ inside the absolute value make constraint $\mathrm{C}_4$ challenging to optimize. To address this, we develop an exact reformulation that eliminates these nonconvexities through variable augmentation while preserving the original solution space.

\begin{proposition} \label{thm:proposition-4}
	Constraint $ \mathrm{C}_{4} $ can be equivalently rewritten as constraints $ \mathrm{G}_{1} $, $ \mathrm{G}_{2} $, $ \mathrm{G}_{3} $, $ \mathrm{G}_{4} $, $ \mathrm{G}_{5} $, $ \mathrm{G}_{6} $, and $ \mathrm{G}_{7} $ 
	\begin{equation} \nonumber
		\mathrm{C}_{4}
		\Leftrightarrow 
			\begin{cases}
				   	\mathrm{G}_{1}: \textstyle \mathbf{d}_{u,k}^\mathrm{H} \mathbf{R} \mathbf{d}_{u,k} \geq \alpha_k \Gamma_\mathrm{th} \sigma_{u}^2, \forall u \in \mathcal{U}, k \in \mathcal{K},
				   	\\
				   	\mathrm{G}_{2}: \left[ \mathbf{R} \right]_{m,m} =  1, \forall m, \in \mathcal{M},
				   	\\
				   	\mathrm{G}_{3}: \left[ \mathbf{R} \right]_{m,m'} =  \left[ \mathbf{R} \right]_{m',m}^*, \forall m, m' \in \mathcal{M}, m'> m,
				   	\\	
				   	\mathrm{G}_{4}: \left[ \mathbf{R} \right]_{m,m'} =  \mathbf{q}^\mathrm{T} \mathbf{Y}_{m,m'} \mathbf{q}^*, \forall m, m' \in \mathcal{M}, m'> m,
				   	\\	
				   	\mathrm{G}_{5}: \mathbf{Y}_{m,m'} \mathbf{1} = \mathbf{z}_m, \forall m, m' \in \mathcal{M}, m' > m,
				   	\\	
				   	\mathrm{G}_{6}: \mathbf{Y}_{m,m'}^\mathrm{T} \mathbf{1} = \mathbf{z}_{m'}, \forall m, m' \in \mathcal{M}, m' > m,
				   	\\	
				   	\mathrm{G}_{7}: \mathbf{Y}_{m,m'} \in \mathbb{B}^{Q \times Q}, \forall m, m' \in \mathcal{M},  m' > m,
			\end{cases}
	\end{equation}
	where $ \mathbf{Y}_{m,m'} $ and $ \mathbf{R} $ are new variables, whereas $ \mathbf{d}_{u,k} = \mathrm{diag} \left( \mathbf{G}^* \mathbf{t}_k^* \right) \mathbf{h}_{u}^* $. 
\end{proposition}

\begin{figure*}[!t]
	\begin{subfigure}[b]{0.32\textwidth}
		\begin{center}
			\includegraphics[height = 3.6cm]{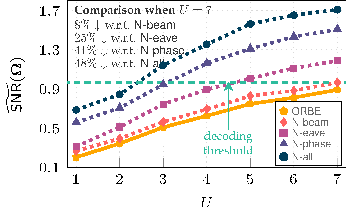}
			\vspace{-6mm}
			\caption{Varying number of eavesdroppers.}
			\label{fig:results-scenario-1}
		\end{center}
	\end{subfigure}
	\hfill 
	\begin{subfigure}[b]{0.32\textwidth}
		\begin{center}
			\includegraphics[height = 3.6cm]{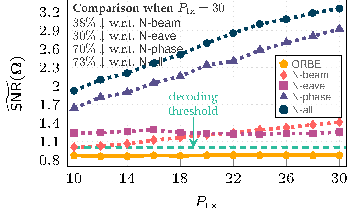}
			\vspace{-6mm}
			\caption{Varying transmit power.}
			\label{fig:results-scenario-2}
		\end{center}
	\end{subfigure}
	\hfill 
	\begin{subfigure}[b]{0.32\textwidth}
		\begin{center}
			\includegraphics[height = 3.6cm]{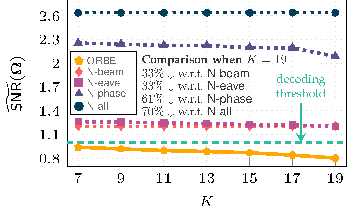}
			\vspace{-6mm}
			\caption{Varying number of beams.}
			\label{fig:results-scenario-4}
		\end{center}
	\end{subfigure}
	\vspace{-1mm}
	\caption{Wiretap SNR under varying numbers of eavesdroppers, numbers of users, and transmit power levels.}
	\vspace{-4mm}
	\label{fig:results-scenario-1-2}
\end{figure*}

\subsection{Reformulated problem} \label{sec:reformulated-problem}

From the above propositions, the problem becomes
\begin{align*} 
	\mathcal{P}': & ~ \underset{\boldsymbol{\Omega}' }{\mathrm{minimize}}  
	& & g \left( \boldsymbol{\Omega}' \right) \triangleq \mu
	\\
	& ~~~~~ \mathrm{s.t.} & & \mathrm{C}_{1} - \mathrm{C}_{2}, \mathrm{D}_{1} - \mathrm{D}_{3}, \mathrm{E}_{1}, \mathrm{F}_{1} - \mathrm{F}_{3}, \mathrm{G}_{1} - \mathrm{G}_{7},
\end{align*}
where $ \boldsymbol{\Omega}' \triangleq \left( \mu, \boldsymbol{\alpha}, \boldsymbol{\delta}, \mathbf{w}, \mathbf{Z}, {\mathbf{R}}, \widebar{\mathbf{Y}} \right) $, $ \boldsymbol{\delta} = \left[ \delta_1, \dots, \delta_E \right]^\mathrm{T} $, $ \mathbf{Z} = \left[ \mathbf{z}_1, \dots, \mathbf{z}_M \right] $, and $ \widebar{\mathbf{Y}} = \left[ \mathbf{Y}_{1,1}, \dots, \mathbf{Y}_{M,M} \right] $.

Our reformulation maps the intractable problem $\mathcal{P}$ into an equivalent \gls{MIQCP}, denoted by $\mathcal{P}'$, replacing the original combinatorial search with a structured optimization problem, which is solved via \glspl{IPM} within a \gls{BnC} framework. The \gls{BnC} framework handles the discrete variables via hierarchical branching, where each node induces a continuous QCP subproblem by relaxing a subset of the binary variables. Since each subproblem is convex, it is solved to global optimality using \gls{IPM}, yielding rigorous dual bounds that enable systematic pruning of the search space via relaxation, bounding, branching, and cutting~\cite{desrosiers2011branch, kronqvist2017center}, ultimately converging to the global optimum while avoiding the prohibitive cost of exhaustive enumeration. Deriving an exact expression for the computational complexity of $\mathcal{P}'$ is challenging due to solver-specific mechanisms within BnC. Nevertheless, an approximate worst-case characterization can be obtained as $\mathcal{O} \left( N_{\mathrm{nodes}} \cdot \sqrt{N_\mathrm{var}} (N_\mathrm{var}^3 + N_\mathrm{var}^2 N_\mathrm{con}) \right)$, where $ N_\mathrm{var} = 1 + K + E + MQ(Q+1) + \frac{M (M-1)}{2} $ and $ N_\mathrm{con} = 3 + E + M + K(2+E+U) + 2MQ(Q+1) + M(M+Q)(M-1) $ denote the number of variables and constraints, respectively, and  $N_{\text{nodes}}$ is the number of explored \gls{BnC} nodes, i.e., the number of convex QCP relaxations solved.


\section{Simulation Results} \label{sec:simulation-results}

Throughout the numerical simulations, large-scale fading coefficients are obtained following the UMa model \cite{3gpp:38.901}, assuming a \gls{BS} antenna gain of $ 10 $ dBi, a user antenna gain of $ 3 $ dBi, and an eavesdropper antenna gain of $ 5 $ dBi. Unless otherwise stated, the parameters are set as $ \chi = \hat{\chi}_u = \check{\chi}_e = 100 $, $ P_\mathrm{tx} = 10 $ W, $ f_0 = 26 $ GHz, $ \sigma_u^2 = \widebar{\sigma}_e^2 = -110 $ dB, $ N = 12 $, $ M = 100 $, and $ M_\mathrm{x} = M_\mathrm{z} = 10 $\footnote{To reflect practical \gls{RIS} deployments with reduced signaling overhead, we employ column-wise control. This yields an azimuthal fan-shaped beam with fixed elevation, aligning with realistic deployment constraints.}. The \gls{BS}-\gls{RIS} distance is $ 10 $ m, with angles $ \gamma^\mathrm{az} \in [20^\circ, 70^\circ] $ and $ \gamma^\mathrm{el} = 90^\circ $. The user-\gls{RIS} distance $ d_u $ ranges from $ 10 $ to $ 20 $ m with angles $ \theta_u^\mathrm{az} \in [110^\circ, 160^\circ] $ and $ \theta_u^\mathrm{el} = 90^\circ $. The eavesdropper-\gls{RIS} distance $ d_e $ ranges from $ 18 $ to $ 40 $ m with angles $ \beta_e^\mathrm{az}  \in [100^\circ, 170^\circ] $ and $ \beta_e^\mathrm{el} = 90^\circ $. A $ 2 $-bit phase control is considered, where the admissible phases are $ \left\lbrace 0^\circ, 90^\circ, 180^\circ, 270^\circ \right\rbrace $. The number of candidate BS-RIS beams is $ K = 15 $, uniformly covering the range $ [10^\circ, 80^\circ] $ (relative to the \gls{RIS}) with a resolution of $ 5^\circ $. All problems are solved using \textsf{CVX} with the \textsf{MOSEK} solver, and the results show the average over $ 100 $ random realizations.

\subsection{Ablation study} 
In \textbf{Scenarios I-III}, we evaluate the impact of neglecting each practical factor in the resource allocation design through an ablation study. The following schemes are considered.

$ \bullet $ {\mytextsf{ORBE}}: \emph{The proposed approach developed in \cref{{sec:proposed-approach}}.}

$ \bullet $ {\mytextsf{N-beam}}: \emph{Assumes a fixed BS-RIS beam, exposing the security inefficiency of static illumination.}

$ \bullet $ {\mytextsf{N-eave}}: \emph{Treats eavesdroppers as non-colluding entities, exposing the ``false sense of security'' caused by neglecting potential collaboration among adversaries.}

$ \bullet $ {\mytextsf{N-phase}}: \emph{Treats RIS phase shifts as continuous during optimization, then projects onto the discrete set, revealing security risks induced by projection.}

$ \bullet $ {\mytextsf{N-all}}: \emph{Combines all the aforementioned assumptions.}

\textbf{Scenario I:} 
\cref{fig:results-scenario-1} illustrates the impact of an increasing number of users, $ U $, with $ E = 5 $ and $ \Gamma_\mathrm{th} = 1 $. As $U$ increases, the colluding eavesdroppers enhance their decoding capability, since a wider distribution of users increases the likelihood of signal interception. In this setting, {\mytextsf{ORBE}} can serve up to seven users without allowing the eavesdroppers to achieve the decoding threshold $ \Gamma_\mathrm{th} = 1 $. In contrast, {\mytextsf{N-beam}} maintains secrecy up to six users, {\mytextsf{N-eave}} up to four, {\mytextsf{N-phase}} up to three, and {\mytextsf{N-all}} up to two. At $ U = 7 $, {\mytextsf{ORBE}} achieves a wiretap \gls{SNR} reduction of  $48\%$ compared to {\mytextsf{N-all}}.

\textbf{Scenario II:}
\cref{fig:results-scenario-2} illustrates the impact of an increasing transmit power, $P_\mathrm{tx}$, with $U = 3$ and $E = 7$. Notably, {\mytextsf{ORBE}} maintains a nearly power-invariant wiretap \gls{SNR}, a direct consequence of the joint optimization between the illumination beam and \gls{RIS} configuration. This framework effectively increases the null-depth at adversarial locations, ensuring that increased transmit power does not translate into proportional leakage. While {\mytextsf{N-eave}} exhibits a similar trend, its wiretap \gls{SNR} remains above the decoding threshold ($\Gamma_\mathrm{th} = 1$) because its non-colluding eavesdropper assumption fails to account for aggregated leakage. Conversely, {\mytextsf{N-beam}} and {\mytextsf{N-phase}} suffer from a rapid, monotonic increase in wiretap \gls{SNR}. For the former, the fixed illumination angle restricts the spatial degrees of freedom, while for the latter, phase quantization errors create leakage floors that scale with $P_\mathrm{tx}$. {\mytextsf{N-all}} suffers the most pronounced degradation due to the combined effect of all impairments. At $P_\mathrm{tx} = 30$~W, {\mytextsf{ORBE}} achieves up to a $73\%$ reduction in wiretap \gls{SNR} compared to {\mytextsf{N-all}}, and is the only scheme that maintains the wiretap \gls{SNR} below $\Gamma_\mathrm{th} = 1$.

\textbf{Scenario III:} \cref{fig:results-scenario-4} illustrates the impact of the number of beams, $K$, with $ U = 3 $ and $ E = 6 $. As $ K $ increases, {\mytextsf{ORBE}} achieves a progressive reduction in the wiretap \gls{SNR}, demonstrating that finer angular granularity in the illumination beam design significantly expands the spatial degrees of freedom for effective null-steering. Conversely, {\mytextsf{N-eave}} exhibits only a marginal response to $K$. This stems from its objective of minimizing the maximum individual wiretap \gls{SNR}, typical of non-colluding eavesdropper models, which fails to exploit beam granularity for aggregate leakage suppression. While {\mytextsf{N-phase}} exhibits a similar downward trend to {\mytextsf{ORBE}} due to improved angular resolution, {\mytextsf{N-beam}} and {\mytextsf{N-all}} remain invariant, as their use of a fixed beam prevents them from exploiting a larger codebook. At $K = 19$, {\mytextsf{ORBE}} achieves up to a $70\%$ reduction in wiretap \gls{SNR} compared with {\mytextsf{N-all}}.

\subsection{Baseline study} 

In \textbf{Scenarios IV-V}, we compare {\mytextsf{ORBE}} against three baselines, which are detailed below.

$ \bullet $ {\mytextsf{BL1}}: \emph{Assumes non-colluding eavesdroppers (e.g., \cite{wang2022multicast, lin2023secure}) with projected phases (e.g., \cite{xu2024reconfigurable, du2025discrete}) and dynamic illumination (e.g., \cite{dejonghe2024design, yilmaz2025joint}). It is optimized via ADMM and supplemented by phase randomization to enforce constraint feasibility.}

$ \bullet $ {\mytextsf{BL2}}: \emph{Assumes non-colluding eavesdroppers (e.g., \cite{wang2022multicast, lin2023secure}), discrete phases (e.g., \cite{abanto2026:fast}), and fixed illumination (e.g., \cite{yan2023passive}). It is optimized via a tailored variant of our framework.}

$ \bullet $ {\mytextsf{BL3}}: \emph{Assumes colluding eavesdroppers (e.g., \cite{shu2023intelligent, wang2022intelligent}) under a fixed illumination beam (e.g., \cite{yan2023passive}) and projected phases (e.g., \cite{xu2024reconfigurable, du2025discrete}). It is optimized via SDR.}

\begin{figure}[!t]
	\begin{subfigure}[b]{0.48\columnwidth}
		\begin{center}
			\includegraphics[height = 3.4cm]{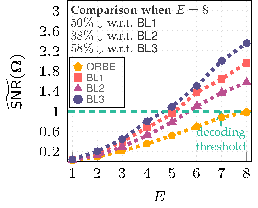}
			\vspace{-6mm}
			\caption{$ Q = 4 $.}
			\label{fig:results-scenario-9a}
		\end{center}
	\end{subfigure}
	\hfill 
	\begin{subfigure}[b]{0.48\columnwidth}
		\begin{center}
			\includegraphics[height = 3.4cm]{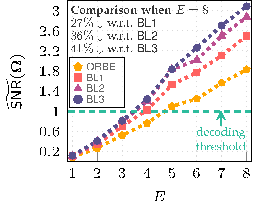}
			\vspace{-6mm}
			\caption{ $ Q = 2 $.}
			\label{fig:results-scenario-9b}
		\end{center}
	\end{subfigure}
	\vspace{-1mm}
	\caption{Wiretap SNR under two quantization levels.}
	\vspace{-4mm}
	\label{fig:results-scenario-3}
\end{figure}
\begin{figure}[!t]
	\begin{subfigure}[b]{0.48\columnwidth}
		\begin{center}
			\includegraphics[]{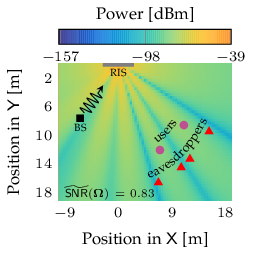}
			\vspace{-6mm}
			\caption{\mytextsf{ORBE}}
			\label{fig:results-scenario-10a}
		\end{center}
	\end{subfigure}
	\hspace{0.2cm} 
	\begin{subfigure}[b]{0.48\columnwidth}
		\begin{center}
			\includegraphics[]{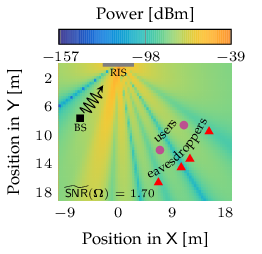}
			\vspace{-6mm}
			\caption{\mytextsf{BL1}}
			\label{fig:results-scenario-10b}
		\end{center}
	\end{subfigure}
	\begin{subfigure}[b]{0.48\columnwidth}
		\begin{center}
			\includegraphics[]{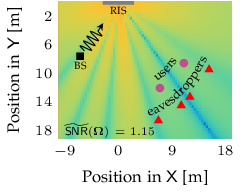}
			\vspace{-2mm}
			\caption{\mytextsf{BL2}}
			\label{fig:results-scenario-10c}
		\end{center}
	\end{subfigure}
	\hspace{0.2cm}
	\begin{subfigure}[b]{0.48\columnwidth}
		\begin{center}
			\includegraphics[]{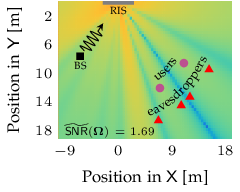}
			\vspace{-2mm}
			\caption{\mytextsf{BL3}}
			\label{fig:results-scenario-10d}
		\end{center}
	\end{subfigure}
	\vspace{-4mm}
	\caption{Received energy heatmaps.}
	\vspace{-4mm}
	\label{fig:results-scenario-10}
\end{figure}
\begin{figure}[!t]
	\begin{center}
		\includegraphics[]{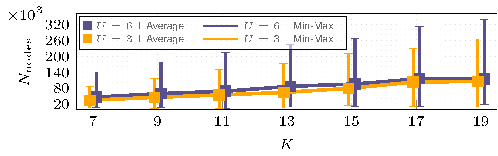}
		\vspace{-2mm}
		\caption{Explored nodes by {\mytextsf{ORBE}}.}
		\vspace{-6mm}
		\label{fig:results-scenario-11}
	\end{center}
\end{figure}

	\textbf{Scenario IV:} 
	\cref{fig:results-scenario-9a} illustrates the impact of an increasing number of eavesdroppers, $ E $, with $ U = 3 $, $ \Gamma_\mathrm{th} = 1 $, and $ Q = 4 $. As $E$ grows, the wiretap \gls{SNR} rises due to the heightened collusive gain. Nevertheless, {\mytextsf{ORBE}} demonstrates strong resilience, maintaining the wiretap \gls{SNR} below the decoding threshold $\Gamma_\mathrm{th} = 1$ for up to eight eavesdroppers, highlighting its robustness in dense adversarial environments. In comparison, {\mytextsf{BL1}} and {\mytextsf{BL2}} maintain secrecy only up to $E=5$, while {\mytextsf{BL3}} fails beyond $ E=4 $. Notably, {\mytextsf{ORBE}} achieves a $ 58\% $ reduction in wiretap \gls{SNR} relative to {\mytextsf{BL3}}. \cref{fig:results-scenario-9b} considers the same setup as \cref{fig:results-scenario-9a} but with a coarser phase resolution of $Q=2$, i.e., phase shifts $\{0^\circ, 180^\circ\}$. This coarser quantization increases the wiretap \gls{SNR} across all schemes and reduces resilience to eavesdropping due to the more constrained phase design space, which limits the available spatial degrees of freedom. Nevertheless, {\mytextsf{ORBE}} maintains a clear advantage, securely tolerating up to $ E=4 $ eavesdroppers, whereas {\mytextsf{BL1}}, {\mytextsf{BL2}}, and {\mytextsf{BL3}} support at most $ E=3 $. This highlights that joint optimization becomes particularly critical under restricted hardware capabilities.

	\textbf{Scenario V:} 
	\cref{fig:results-scenario-10} illustrates the received energy heatmaps in a rectangular setting with $ \Gamma_\mathrm{th} = 1 $. Relative to the RIS, the BS is located at $\gamma^\mathrm{az} = 50^\circ$, the users at $\boldsymbol{\theta}^\mathrm{az} = [120^\circ, 142^\circ]$, and the eavesdroppers at $\boldsymbol{\beta}^\mathrm{az} = [112^\circ, 126^\circ, 132^\circ, 148^\circ]$, with distances fixed at $\hat{d}_u = 14$ m and $\check{d}_e = 17$ m. {\mytextsf{ORBE}} achieves $\widetilde{\mathsf{SNR}}(\boldsymbol{\Omega}) = 0.83$, outperforming {\mytextsf{BL1}} ($1.70$), {\mytextsf{BL2}} ($1.15$), and {\mytextsf{BL3}} ($1.69$). This gain stems from its effective and spatially consistent energy suppression, as evidenced by the pronounced low-energy regions (deep blue) aligned with the eavesdropper directions. In contrast, {\mytextsf{BL1}}, {\mytextsf{BL2}}, and {\mytextsf{BL3}} exhibit less consistent spatial suppression, with localized nulls at some eavesdropper directions but elevated leakage at others, resulting in higher overall security risk. Notably, {\mytextsf{ORBE}} is the only scheme that maintains the wiretap \gls{SNR} below the decoding threshold ($\Gamma_\mathrm{th} = 1$). This underscores the importance of its joint design and globally optimal formulation in reliably safeguarding multicast transmissions.

\subsection{Complexity reduction}

This scenario evaluates the complexity reduction achieved by {\mytextsf{ORBE}} relative to exhaustive search, enabled by the \gls{BnC}.

	\textbf{Scenario VI:} \cref{fig:results-scenario-11} shows the number of explored nodes for {\mytextsf{ORBE}} as a function of $K$, following the settings of Scenario~III for $U = 3$ and $U = 6$. The number of explored nodes, $N_{\mathrm{nodes}}$ (as defined in the computational complexity analysis of $\mathcal{P}'$ in \cref{sec:reformulated-problem}), increases with $K$, exhibiting an approximately linear trend on average. Importantly, the number of explored nodes remains orders of magnitude lower than that required by exhaustive search, which spans from $7.3 \times 10^6$ candidate solutions at $K=7$ to $19.9 \times 10^6$ at $K=19$. In contrast, {\mytextsf{ORBE}} evaluates less than $1\%$ of the potential candidates on average, and remains below $2\%$ in the worst-case scenario. These results validate that the proposed \gls{MIQCP} formulation effectively prunes vast regions of the search space, ensuring computational tractability and global optimality with significantly reduced overhead.


\section{Conclusions} \label{sec:conclusions}

This work showed that widely used assumptions, such as non-colluding eavesdroppers, fixed BS-RIS illumination, and phase relaxation followed by projection, can severely underestimate secrecy risks in RIS-assisted multicasting. We proposed {\mytextsf{ORBE}}, a globally optimal framework that jointly optimizes the BS beam and discrete RIS phase shifts while accounting for colluding eavesdroppers. Results confirm that properly modeling these practical factors is critical for suppressing wiretap SNR and ensuring secure RIS-assisted multicasting.  
 

\section*{Acknowledgment} \label{sec:acknowledgment}
This research was supported by the Federal Ministry of Research, Technology and Space (BMFTR) under Grants 16KIS2411.

\appendix

\centerline{\textbf{Proof of Proposition 1}}
We introduce a new variable $ \mu $, specified as constraint $ \mathrm{D}_{1}: \mu \in \mathbb{R}_+ $, to bound $ f \left( \boldsymbol{\Omega} \right) $. The objective is then expressed as $ g \left( \boldsymbol{\Omega}' \right) \triangleq \mu $ by introducing constraint $ \mathrm{D}_\mathrm{aux,1}: \widetilde{\mathsf{SNR}} \left( \boldsymbol{\Omega} \right) \leq \mu $, where $ \boldsymbol{\Omega}' $ denotes the decision variables in $ \mathcal{P}'$. To bound each eavesdropper's SNR individually, we introduce variables $\delta_e$ with $\mathrm{D}_{2}: \delta_e \in \mathbb{R}+$, $\forall e \in \mathcal{E}$. Constraint $\mathrm{D}_{\mathrm{aux},1}$ can then be equivalently decomposed into $ \mathrm{D}_{3}: \sum_{e \in \mathcal{E}} \delta_{e}^2 \leq \mu $, and $ \mathrm{D}_\mathrm{aux,2}: \widebar{\mathsf{SNR}}_{e} \left( \boldsymbol{\Omega} \right) \leq \delta_{e}^2, \forall e \in \mathcal{E} $.	By defining $ \mathbf{p}_{e,k} = \mathrm{diag} \left( \mathbf{G}^* \mathbf{t}_k^* \right) \mathbf{f}_{e}^* $, we can transform $ \mathrm{D}_\mathrm{aux,2} $ into $ \mathrm{D}_\mathrm{aux,3}: \big| \textstyle \sum_{k \in \mathcal{K}} \alpha_k \mathbf{p}_{e,k}^\mathrm{H} \mathbf{w} \big|^2 / \widebar{\sigma}_{e}^2 \leq \delta_{e}^2, \forall e \in \mathcal{E} $. Since both sides of $ \mathrm{D}_\mathrm{aux,3} $ are nonnegative, we take the square root, yielding $ \mathrm{D}_{4}: \big| \textstyle \sum_{k \in \mathcal{K}} \alpha_k \mathbf{p}_{e,k}^\mathrm{H} \mathbf{w} \big| / \widebar{\sigma}_{e} \leq \delta_{e}, \forall e \in \mathcal{E} $.

\vspace{1.5mm}
\centerline{\textbf{Proof of Proposition 2}}
For any $ \mathsf{E}_e $, applying Jensen's inequality to the quadratic term on left-hand side of $ \mathrm{D}_{4} $ yields $ \mathrm{E}_\mathrm{aux,1}: \big| \textstyle \sum_{k \in \mathcal{K}} \alpha_k \mathbf{p}_{e,k}^\mathrm{H} \mathbf{w} \big| \leq \textstyle \textstyle \sum_{k \in \mathcal{K}} \big| \alpha_k \mathbf{p}_{e,k}^\mathrm{H} \mathbf{w} \big| $.
Let $k'$ denote the index of the selected transmit beam $\mathbf{t}_{k'}$. Then, $\mathrm{E}_{\mathrm{aux},1}$ can be expressed as $ \mathrm{E}_\mathrm{aux,2}: \big| \textstyle \alpha_{k'} \mathbf{p}_{e,k'}^\mathrm{H} \mathbf{w} + \sum_{k \neq k'} \alpha_k \mathbf{p}_{e,k}^\mathrm{H} \mathbf{w} \big| \leq \big| \alpha_{k'} \mathbf{p}_{e,k'}^\mathrm{H} \mathbf{w} \big| + \textstyle \sum_{k \neq k'} \big| \alpha_k \mathbf{p}_{e,k}^\mathrm{H} \mathbf{w} \big| $.
Since $\alpha_k = 0$ for all $k \neq k'$, the inequality becomes tight. Consequently, $\mathrm{D}_{4}$ can be equivalently rewritten as $ \mathrm{E}_\mathrm{aux,3}: \textstyle \sum_{k \in \mathcal{K}} \big| \alpha_k  \mathbf{p}_{e,k}^\mathrm{H} \mathbf{w} \big| / \widebar{\sigma}_{e} \leq \delta_{e}, \forall e \in \mathcal{E} $. Using the binary nature of $\alpha_k$, this simplifies to $ \mathrm{E}_\mathrm{aux,4}: \textstyle \sum_{k \in \mathcal{K}} \alpha_k \big|  \mathbf{p}_{e,k}^\mathrm{H} \mathbf{w} \big| / \widebar{\sigma}_{e} \leq \delta_{e}, \forall e \in \mathcal{E} $.
Since only one beam is selected at the \gls{BS}, the summation in $\mathrm{E}_{\mathrm{aux},4}$ can be decoupled using the \emph{big-M} method. This removes the bilinear coupling between $\alpha_k$ and $\mathbf{w}$, yielding $ \mathrm{E}_{1}: \big|  \mathbf{p}_{e,k}^\mathrm{H} \mathbf{w} \big| / \widebar{\sigma}_{e} \leq \delta_{e} + (1 - \alpha_k) L_{e}, \forall e \in \mathcal{E}, k \in \mathcal{K} $, where $ L_{e} = \sqrt{M P_\mathrm{tx}} \left\| \mathbf{f}_{e} \right\|_2 \left\| \mathbf{G} \right\|_\mathrm{F} / \widebar{\sigma}_{e} $ is an upper bound for $\big|  \mathbf{p}_{e,k}^\mathrm{H} \mathbf{w} \big| / \widebar{\sigma}_{e} $, computed via the Cauchy-Schwarz inequality.

\vspace{1.5mm}
\centerline{\textbf{Proof of Proposition 3}}
Following the one-hot encoding method in \cite{abanto2026:fast}, we represent the phase selection for each \gls{RIS} element by introducing the binary vectors $ \mathrm{F}_1: \mathbf{z}_m \in \mathbb{B}^{Q \times 1}, \forall m \in \mathcal{M} $. To ensure physical consistency, we enforce $ \mathrm{F}_2: \mathbf{1}^\mathrm{T} \mathbf{z}_m = 1, \forall m \in \mathcal{M} $, which guarantees that exactly one discrete phase shift is selected per element. The mapping between the admissible phase shifts, collected in $\mathbf{q} = [ \mathrm{e}^{\mathrm{j} \phi_1}, \dots, \mathrm{e}^{\mathrm{j} \phi_Q} ]^\mathrm{T}$, and the \gls{RIS} configuration is established via $\mathrm{F}_3: [\mathbf{w}]_m = \mathbf{q}^\mathrm{T} \mathbf{z}_m, \forall m \in \mathcal{M}$.

\vspace{1.5mm}
\centerline{\textbf{Proof of Proposition 4}}
Note that $ \mathrm{C}_{4} $ is equivalent to constraint $ \mathrm{G}_\mathrm{aux,1}: \textstyle \big| \sum_{k \in \mathcal{K}} \alpha_k \mathbf{d}_{u,k}^\mathrm{H} \mathbf{w} \big|^2 \geq \Gamma_\mathrm{th} \sigma_{u}^2 , \forall u \in \mathcal{U} $, where $ \mathbf{d}_{u,k} = \mathrm{diag} \left( \mathbf{G}^* \mathbf{t}_k^* \right) \mathbf{h}_{u}^* $. By applying Jensen's inequality and following the procedure in the proof of \textbf{\cref{thm:proposition-2}}, $ \mathrm{G}_\mathrm{aux,1} $ reduces to $ \mathrm{G}_\mathrm{aux,2}: \textstyle \sum_{k \in \mathcal{K}} \alpha_k \big| \mathbf{d}_{u,k}^\mathrm{H} \mathbf{w} \big|^2 \geq \Gamma_\mathrm{th} \sigma_{u}^2 , \forall u \in \mathcal{U} $. Since only one variable $ \alpha_k $ is equal to $ 1 $, we can decompose $ \mathrm{G}_\mathrm{aux,2} $ into the intersection of multiple constraints, collectively defined as $ \mathrm{G}_\mathrm{aux,3}: \textstyle \big| \mathbf{d}_{u,k}^\mathrm{H} \mathbf{w} \big|^2 \geq \alpha_k \Gamma_\mathrm{th} \sigma_{u}^2, \forall u \in \mathcal{U}, k \in \mathcal{K} $. Furthermore, $ \mathrm{G}_\mathrm{aux,3} $ can be recast as $ \mathrm{G}_{1}: \textstyle \mathbf{d}_{u,k}^\mathrm{H} \mathbf{R} \mathbf{d}_{u,k} \geq \alpha_k \Gamma_\mathrm{th} \sigma_{u}^2, \forall u \in \mathcal{U}, k \in \mathcal{K} $, subject to including a new constraint $ \mathrm{G}_\mathrm{aux,4}: \mathbf{R} = \mathbf{w} \mathbf{w}^\mathrm{H} $, as in \cite{abanto2025:optimal-user-target-scheduling-user-target-pairing-low-resolution-phase-only-beamforming-isac-systems}. Adopting an entry-wise notation, $ \mathrm{G}_\mathrm{aux,4} $ is equivalent to $ \mathrm{G}_{\mathrm{aux},5}: \left[ \mathbf{R} \right]_{m,m'} = \left[ \mathbf{w} \right]_m \left[ \mathbf{w}^{*} \right]_{m'}, \forall m, m' \in \mathcal{M} $. Leveraging $ \mathrm{F}_{3} $, we further transform $ \mathrm{G}_{\mathrm{aux},5} $ into $ \mathrm{G}_{\mathrm{aux},6}:  \left[ \mathbf{R} \right]_{m,m'} = \left( \mathbf{q}^\mathrm{T} \mathbf{z}_m \right) \left( \mathbf{q}^\mathrm{T} \mathbf{z}_{m'} \right)^* = \left( \mathbf{q}^\mathrm{T} \mathbf{z}_m \right) \left( \mathbf{z}_{m'}^\mathrm{T} \mathbf{q} \right)^* = \mathbf{q}^\mathrm{T} \mathbf{z}_m \mathbf{z}_{m'}^\mathrm{T} \mathbf{q}^* $. By exploiting the conjugate symmetry of $ \mathbf{R} $, we reformulate $ \mathrm{G}_{\mathrm{aux},6} $ as the following three constraints: $ \mathrm{G}_{2}: \left[ \mathbf{R} \right]_{m,m} =  1, \forall m, \in \mathcal{M} $, $ \mathrm{G}_{3}: \left[ \mathbf{R} \right]_{m,m'} =  \left[ \mathbf{R} \right]_{m',m}^*, \forall m, m' \in \mathcal{M}, m'> m $, and $ \mathrm{G}_{4}: \left[ \mathbf{R} \right]_{m,m'} =  \mathbf{q}^\mathrm{T} \mathbf{Y}_{m,m'} \mathbf{q}^*, \forall m, m' \in \mathcal{M}, m'> m $, subject to introducing $ \mathrm{G}_{\mathrm{aux},7}: \mathbf{Y}_{m,m'} = \mathbf{z}_{m} \mathbf{z}_{m'}^\mathrm{T}, \forall m, m' \in \mathcal{M}, m' > m $. Note that $ \mathrm{G}_{\mathrm{aux},7} $ couples $ \mathbf{z}_m $ and $ \mathbf{z}_{m'} $ multiplicatively, complicating tractability. To address this, we multiply both sides of $ \mathrm{G}_{\mathrm{aux},7} $ by $ \mathbf{1} $, yielding $ \mathrm{G}_{\mathrm{aux},8}: \mathbf{Y}_{m,m'} \mathbf{1} = \mathbf{z}_m \mathbf{z}_{m'}^\mathrm{T} \mathbf{1}, \forall m, m' \in \mathcal{M}, m' > m $. Leveraging $ \mathrm{F}_{2} $, which states that $ \mathbf{z}^\mathrm{T} \mathbf{1} = 1 $, then $ \mathrm{G}_{\mathrm{aux},8} $ reduces to $ \mathrm{G}_{5}: \mathbf{Y}_{m,m'} \mathbf{1} = \mathbf{z}_m, \forall m, m' \in \mathcal{M}, m' > m $. Similarly, from $ \mathrm{G}_{\mathrm{aux},7} $ we can also obtain $ \mathrm{G}_{6}: \mathbf{Y}_{m,m'}^\mathrm{T} \mathbf{1} = \mathbf{z}_{m'}, \forall m, m' \in \mathcal{M}, m' > m $. Given that $ \mathbf{Y}_{m,m'} $ yields from multiplying $ \mathbf{z}_n $ and $ \mathbf{z}_m $, each of which has one element $ 1 $, then $ \mathbf{Y}_{m,m'} $ is binary. This conditions is enforced through $ \mathrm{G}_{7}: \mathbf{Y}_{m,m'} \in \mathbb{B}^{Q \times Q}, \forall m, m' \in \mathcal{M}, m' > m $.

\bibliographystyle{IEEEtran}
\bibliography{IEEEabrv,ref}


\end{document}